\documentclass[12pt,a4paper]{iopart}

\usepackage{iopams} 
\usepackage{amsthm}
\usepackage{graphicx}
\usepackage{enumerate} 
\usepackage{microtype}
\usepackage[dvipsnames]{xcolor}
\usepackage[colorlinks,allcolors=NavyBlue]{hyperref}
\usepackage{doi}
\usepackage{booktabs}
\usepackage[group-digits=false]{siunitx}
\usepackage{nowidow}

\DeclareRobustCommand\binom[2]{{\begingroup #1\endgroup \choose #2}}

\newcommand*{\dmax}{\ensuremath{d_\mathrm{max}}}
\newcommand*{\myref}[2]{\hyperref[#2]{#1~\ref*{#2}}}
\newcommand*{\myrefp}[3]{\hyperref[#2]{#1~\ref*{#2}#3}}

\theoremstyle{plain}
\newtheorem{thm}{Theorem}
\newtheorem{lem}[thm]{Lemma}

\theoremstyle{definition}
\newtheorem{defn}{Definition}

\begin{document}

\title
[Connectedness matters: Sampling connected graphs with a given degree sequence]
{Connectedness matters: Construction and exact random sampling of connected networks}

\author{Szabolcs~Horv\'{a}t$^{1,2,3}$ and Carl~D~Modes$^{1,2,4}$}

\address{$^1$ 
Max Planck Institute for Cell Biology and Genetics,
Pfotenhauerstra\ss{}e~108, 01307~Dresden, Germany}

\address{$^2$ 
Center for Systems Biology Dresden,
Pfotenhauerstra\ss{}e~108, 01307~Dresden, Germany}

\address{$^3$ 
Max Planck Institute for the Physics of Complex Systems,
N\"othnitzerstra\ss{}e~38, 01187~Dresden, Germany}

\address{$^4$ 
Physics of Life, Biotechnology Center of the TU Dresden,
Tatzberg~47/49, 01307~Dresden, Germany}

\eads{horvat@mpi-cbg.de, modes@mpi-cbg.de}
\vspace{10pt}
\begin{indented}
\item[]30 November 2020
\end{indented}

\begin{abstract}
We describe a new method for the random sampling of connected networks with a specified degree sequence. We consider both the case of simple graphs and that of loopless multigraphs. The constraints of fixed degrees and of connectedness are two of the most commonly needed ones when constructing null models for the practical analysis of physical or biological networks. Yet handling these constraints, let alone combining them, is non-trivial. Our method builds on a recently introduced novel sampling approach that constructs graphs with given degrees independently (unlike edge-switching Markov Chain Monte Carlo methods) and efficiently (unlike the configuration model), and extends it to incorporate the constraint of connectedness. Additionally, we present a simple and elegant algorithm for directly constructing a single connected realization of a degree sequence, either as a simple graph or a multigraph. Finally, we demonstrate our sampling method on a realistic scale-free example, as well as on degree sequences of connected real-world networks, and show that enforcing connectedness can significantly alter the properties of sampled networks.
\end{abstract}

\vspace{2pc}
\noindent{\it Keywords}: graph theory, network science, null models, random sampling, degree sequence, connectedness, graph construction, algorithms.
\submitto{{\it J.\ Phys.:\ Complexity}}
%
%
%

\section{Introduction}

From the active scaffolding of actomyosin in the cell's cortex to the underlying gene expression machinery that regulates it, from the neighbourhood interactions of grains in a sand pile to those of the engineered struts and cables in a suspension bridge, and from the flow of virtual traffic on the internet to, critically in the time of COVID-19, the web of 
contacts that allow the spread of viral contagion, network structures underlie the vast majority of sufficiently complex real-world systems. Unsurprisingly, then, a great deal of focus has been placed on the furtherance of our understanding of how these network structures affect and ultimately determine the physical, biological, and social phenomena that play out over them. Indeed, the explosive growth of the fields of network science and complexity science in the last two decades is a direct consequence of this focus.

As is to be expected in such a young field, however, there remain fundamental challenges. One such challenge is the surprising difficulty of translating the simple concept of the null hypothesis into a network setting. Done directly, such a translation would read: ``There is no relationship between the network structure or properties and the observed or measured phenomena of interest.'' But of course one cannot simply compare the case of phenomena potentially arising from some specific network structure with a case of \emph{no network at all}, forcing one to conclude that the correct operational statement of the null hypothesis in the complex network milieu must be: ``There would be no difference in the observed phenomena or measured output if the specific underlying network were to be replaced by a generic network.'' And herein lies the rub. What is a \emph{generic} network? Surely one can demand that the generic network---or ensemble of generic networks---satisfy some small set of constraints in order to ensure relevance to the biology, physics, or social dynamics under consideration. For example, in epidemiological viral spreading models it would be of no use to consider a heavily disconnected network with many small individual components to be among the generic networks. In fact, it is the unfortunate state of affairs that this simple issue is so tricky that many network and complexity science results are reported and accepted without reference to a test of the null hypothesis! But before we can fruitfully return to the question of membership amongst the relevant generic networks, we must first briefly discuss the problem of sampling from constrained ensembles of networks.

Indeed, the so-called \emph{random graph models} are among the most powerful tools of network science. Essentially, a random graph model is simply a probability distribution defined over a set $\mathcal{G}$ of graphs, also referred to as a graph ensemble. Often, such models are defined through an explicit stochastic graph construction process: the Watts--Strogatz model \cite{Watts1998} and the preferential attachment model \cite{Barabasi1999,Bollobas2001} are some well-known examples. Usually, such constructive models are introduced and studied because the graphs they produce have some interesting or relevant property: The Watts--Strogatz model can produce graphs with the ``small-world'' property, which is famously present in social networks. The preferential attachment model can produce ``scale-free'' graphs, i.e.\ graphs with a power-law degree distribution, a much-studied property which many real-world networks possess \cite{Voitalov2019,Broido2019}. However, not all scale-free networks can be produced by the preferential attachment mechanism, and one cannot make general statements about all scale-free networks based only on those generated by a preferential attachment model. Therefore, for some purposes, it is useful to define random graph models not through a construction process, but by directly imposing a property of interest. The simplest way to define such a model (i.e.\ distribution) is to constrain its support $\mathcal{G}$ to include only those graphs that possess a given property, and assign the same probability to all elements of~$\mathcal{G}$.  The graph ensemble obtained this way represents the property of interest the best. A related approach constrains the averages of some numerical graph properties and defines the distribution over $\mathcal{G}$ to be the one with maximal entropy, which leads to \emph{exponential random graph models} \cite{Holland1981,Park2004,Horvat2015}.

Returning to the challenge of rendering the null hypothesis in a network setting, constraint-based random graph models are particularly useful as \emph{null models}. Null models are used to determine if an interesting observed feature of some empirical network can be explained by another simpler feature. The simpler feature is used as a constraint to define a random graph model, which is then compared to the empirical data. Another application is dealing with incomplete empirical data. Sometimes, it is not possible to fully map the connections of a real-world network, either due to practical limitations or, in the case of networks of people, due to privacy concerns \cite{Liljeros2001}. In such cases, the known data can be incorporated as a constraint into a random graph model, and individual networks sampled from the model can be used as proxies for the (unknown) real network. Both applications require being able to computationally generate samples from the model. In the case of constraint-based models this means restricting the set of graphs $\mathcal{G}$ to only those that satisfy the constraint, then performing uniform sampling. This is usually a difficult problem, as there are no general sampling methods that work with arbitrary constraints. Each constraint requires developing a sampling algorithm specific to it, and combining multiple constraints is a significant additional challenge.

In this paper we consider the problem of sampling \emph{connected} graphs with \emph{a given degree sequence}. Constraining the \emph{degrees} has countless practical applications: It is a frequently used null model, for example when finding network motifs \cite{Milo2002}, detecting a so-called ``rich-club structure'' \cite{Colizza2006} or analysing the assortative structure of networks \cite{Newman2002}. Degree-constrained random graphs are also useful as proxies when only the degrees of an empirical network are known, such as in the case of the famous web of sexual connections dataset \cite{Liljeros2001}. The constraint of \emph{connectedness} is a frequent additional requirement: Many real-world networks, such as vasculature, brain networks, or molecules (networks of atoms), are always connected. Commonly used network measures like closeness centrality are only meaningful for connected graphs. Processes such as epidemic spreading must be modelled on connected networks. In this work we present a novel method to handle these two constraints, \emph{degrees} and \emph{connectedness}, simultaneously.

The article is organized as follows: \myref{Section}{sec:foundation} introduces the mathematical background used in later parts of the article, and reviews existing sampling methods for graphs with constrained degrees. \myref{Section}{sec:connected} presents a new and simple algorithm to construct a single connected graph with given degrees. \myref{Section}{sec:biased} presents a recently introduced family of importance sampling methods for graphs with constrained degrees, and shows how to incorporate the additional constraint of connectedness. Finally, \myref{section}{sec:numerics} demonstrates the practical applicability of the method on both synthetic and real-world examples.

\section{Mathematical foundations}
\label{sec:foundation}

In this section we introduce the concepts, terminology and notations used in the rest of the work. We say that a graph is \emph{simple} if it has no multi-edges or self-loops, i.e.\ if any two vertices have at most one connection between them, and no vertex is connected to itself. The \emph{degree} $d$ of a vertex is the number of connections it has. The \emph{degree sequence} of a graph on $n$ vertices is simply the collection of its vertex degrees, $\mathbf{d} = (d_1, d_2, \dots, d_n)$. If the degree sequence of a graph $G$ is $\mathbf{d} = (d_1, d_2, \dots, d_n)$, we say that \emph{the graph $G$ realizes the degree sequence $\mathbf{d}$}.

Since each edge in a graph connects to a vertex at both of its endpoints, the sum of the degrees in a graph is twice the number of its edges, an even number. This statement is commonly known as \emph{the handshaking lemma}. But not every even-sum sequence of integers is realizable as a simple graph. For example, $\mathbf{d} = (3,2,1)$ can only be realized by either a graph that includes self-loops,
\raisebox{-0.25\height}[0pt][0pt]{\includegraphics{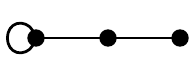}}, 
or one that includes multi-edges, \raisebox{-0.25\height}[0pt][0pt]{\includegraphics{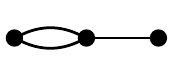}}.

\begin{defn}[graphicality]
A degree sequence is said to be \textbf{graphical} if there is a simple graph that realizes it.
\end{defn}

\noindent
The well-known Erd\H{os}--Gallai theorem provides a direct way to check if a degree sequence is graphical.

\begin{thm}[Erd\H{o}s and Gallai, \cite{ErdosGallai1960}]
\label{thm:eg}
Let $d_1 \ge d_2 \ge \cdots \ge d_n$ be a degree sequence. There is a simple graph that realizes this degree sequence if and only if $\sum_{i=1}^n d_i$ is even and 
\begin{equation}
\sum_{i=1}^k d_i \le k(k-1) + \sum_{i=k+1}^{n} \min(d_i, k)
\end{equation}
for every $1 \le k \le n$.
\end{thm}

\noindent 
Tripathi and Vijay have shown that it is sufficient to check these inequalities for those $k$ where $d_k \ne d_{k+1}$ and for $k=n$ \cite{Tripathi2003}. Using this stricter version of the theorem, it is possible to perform the checks in linear computational time. Kir\'{a}ly \cite{Kiraly2012} and Cloteaux  \cite{Cloteaux2016} describe two such linear-time algorithms for testing graphicality.

\begin{defn}[multigraphicality]
A degree sequence is said to be \textbf{multigraphical} if there is a graph, potentially containing multi-edges, but no self-loops, that realizes it. We refer to such a graph as a \emph{loopless multigraph}.
\end{defn}

\begin{thm}[multigraphicality]
\label{thm:multigraphical}
Let $\mathbf{d} = (d_1, d_2, \dots, d_n)$ be a degree sequence. There is a loopless multigraph that realizes $\mathbf{d}$ if and only if $\sum_{i=1}^n d_i$ is even and
\begin{equation}
\frac{1}{2} \sum_{i=1}^n d_i \ge \dmax,
\label{eq:multigraphical}
\end{equation}
where \dmax\ denotes the largest degree in $\mathbf{d}$.
\end{thm}

\noindent
The proof of \myref{theorem}{thm:multigraphical} is given in \ref{apd:multigraphical}.

Not every graphical or multigraphical degree sequence has a \emph{connected} realization. For example, the degree sequence $(1,1,1,1)$ is only realized by the non-connected graph \raisebox{-0.25\height}[0pt][0pt]{\includegraphics{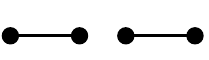}}.

\begin{defn}[potential connectedness]
A degree sequence is said to be \textbf{potentially connected} if it has a realization that is connected.
\end{defn}

\noindent
The concept of potential connectedness also applies to degree sequences which only have non-simple realizations. However, it can be shown that all potentially connected sequences that are graphical have connected realizations that are also simple.

In this paper we consider so-called \textbf{labelled graphs}, i.e.\ we consider the vertices to be distinguishable. Thus, the degree sequence $(d_1, d_2, d_3, d_4) = (1,2,2,1)$ is taken to have two isomorphic but distinct realizations as \raisebox{-0.333\height}[0pt][0pt]{\includegraphics{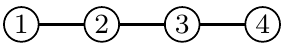}} and \raisebox{-0.333\height}[0pt][0pt]{\includegraphics{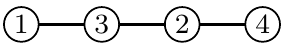}}.

\subsection{Approaches to sampling graphs with a given degree sequence}

There are two widely used approaches to uniformly sampling simple labelled graphs with a prescribed degree sequence: (1)~``stub-matching'' algorithms such as the configuration model and (2)~Markov chain Monte Carlo sampling based on degree-preserving edge switches. We briefly review both families of methods, and consider how the additional constraint of connectedness can be incorporated into them.

\textbf{The configuration model}, also called the \emph{pairing model}, is probably the simplest and most widely known approach to generating random graphs with a given degree sequence.  The sampling algorithm proceeds as follows: Let us consider each vertex with as many unconnected stubs as its degree, as shown in \myref{figure}{fig:stubs}. Then repeatedly pick two not-yet-connected stubs uniformly at random and connect them, until there are no unconnected stubs left. This algorithm may clearly produce graphs that are not simple (i.e.\ they have multi-edges or self-loops). Such graphs are simply rejected, and the generation procedure is restarted. 

The configuration model's algorithm produces each simple realization of the degree sequence with the same probability (although the same is not true for non-simple ones) \cite{NewmanBook}. Therefore, by rejecting the non-simple outcomes, the simple realizations can be sampled uniformly. It is important to note that if the outcome is non-simple, the generation procedure must be restarted from the beginning. It is not sufficient to merely reject any stub pairing that creates a non-simple edge and choose another one instead. Doing so would no longer produce each realization with the same probability, as is shown in \myref{section}{sec:biased}.

\begin{figure}[tb!]
   \centering
   \includegraphics{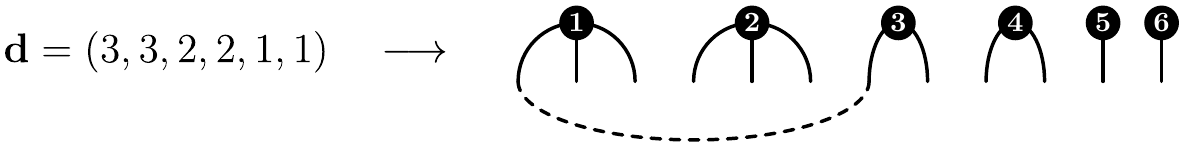} 
   \caption{A degree sequence can be visualized by sketching each vertex with as many unconnected stubs as its degree. Graph construction proceeds by connecting the stubs (dashed line).}
   \label{fig:stubs}
\end{figure}

The configuration model works well for sparse graphs that have small degrees. However, as the graph gets denser, 
the probability that the algorithm generates a non-simple graph, which must be rejected, increases quickly. For dense graphs, the rejection rate becomes too high for this sampling method to be computationally feasible. The same is true for degree sequences of sparse graphs that have a few very high degree vertices, such as scale-free and other heavy tail degree distributions, which are commonly observed in real-world networks \cite{Voitalov2019,Broido2019}. Therefore, the configuration model is only practical in some limited situations.

The constraint of connectedness can be incorporated trivially into the configuration model: simply reject any non-connected outcomes along with the non-simple ones. However, usually, most realizations of a sparse degree sequence are not connected, increasing the rejection rate further. This makes the connected version of the configuration model unfeasible for sparse graphs as well.

\textbf{Edge-switching Markov chain Monte Carlo (MCMC) methods} work by first building a single realization of the degree sequence, then repeatedly modifying the graph using random, degree sequence preserving \emph{edge switches} like those shown in \myref{figure}{fig:edge-switch}. It can be shown that even though not all pairs of edges can be switched without creating a non-simple graph, all simple realizations of a degree sequence can be reached with permissible edge switches. Consequently, a Markov chain constructed using edge switches is irreducible. It can be shown that if the edges to be switched are chosen uniformly at random, and the switch is simply not performed when it would create a multi-edge, then the stationary distribution of the Markov chain will be uniform. Details are given in \ref{apd:mcmc}. Sampling can be performed as usual with MCMC, by recording states from the chain at certain intervals.

\begin{figure}[tb!]
   \centering
   \includegraphics{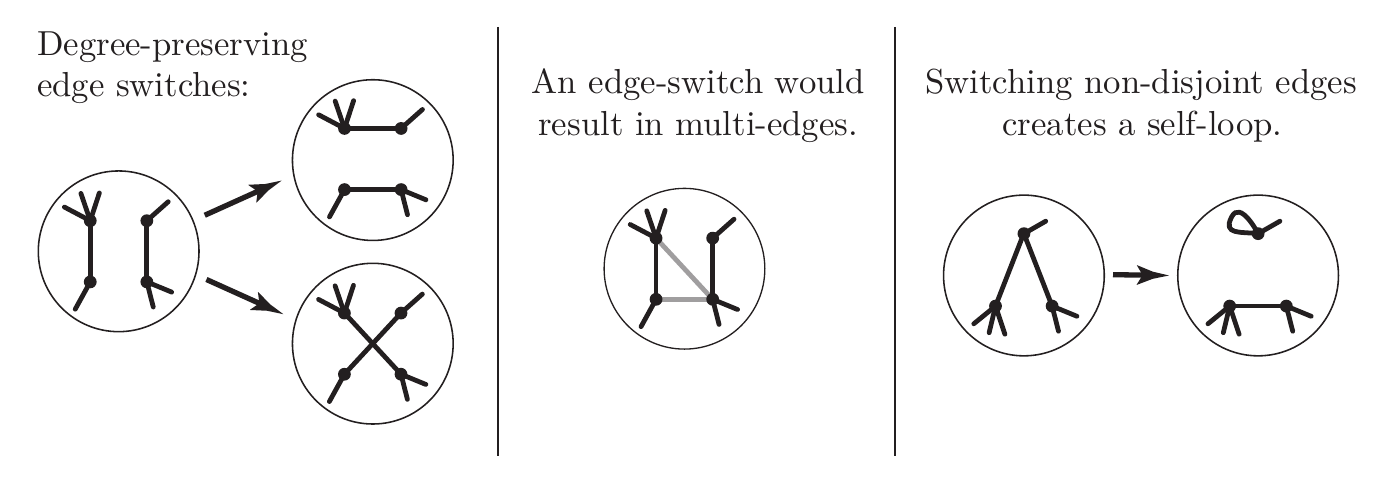} 
   \caption{\emph{Edge switches} can be used to modify a graph while preserving its degree sequence, as follows: 1.\ select two edges $(a,b)$ and $(c,d)$; 2.\ remove them and add either $(a,c),\, (b,d)$, or $(a,d),\, (b,c)$ instead. Not all pairs of edges can be switched without creating multi-edges. If the two edges share a vertex, then the switch creates a self-loop.}
   \label{fig:edge-switch}
\end{figure}

Incorporating the connectedness constraint into such a sampler is more involved than in the case of the configuration model. The Markov chain is still irreducible if edge switches that would disconnect the graph are forbidden \cite{Taylor1981}. However, testing whether an edge switch disconnects the graph takes computational time proportional to the size of the graph. Performing this test after every edge switch would make the method impractically slow.  While there are published algorithms that make use of information from previous connectedness tests to achieve an average polylogarithmic complexity when a series of incremental changes are made to the graph \cite{Henzinger1999,Thorup2000,Holm2001}, these algorithms are complicated and their implementation is involved. It is unclear if they would perform sufficiently well in practice. We are not aware of any MCMC-based graph sampler implementation that makes use of them. More practical approaches perform multiple edge switches between connectedness checks \cite{Gkantsidis2003,Viger2015}. Frequent connectedness checks would result in bad computational performance, while an insufficient number of checks makes it more likely that the graph becomes disconnected, and therefore the last few edge switches must be reverted. These methods use heuristics to find an optimal number of switches to perform between connectedness checks, and maximize performance.

An alternative approach uses only local edge switches performed between pairs of edges that are connected by a third edge. These restricted switches, called \emph{edge flips}, preserve connectedness. However, the flip Markov chain is irreducible only in certain special cases, such as for regular graphs or when all realizations of the degree sequence have a diameter of at least 4 \cite{Cooper2019,Feder2006,Mahlmann2005}. Furthermore, it requires performing a larger number of moves to randomize the graph than the other MCMC approaches.

The disadvantage of MCMC-based methods is that the mixing time of these Markov chains is not known in general \cite{Milo2003,Erdos2019}. Therefore, one must use heuristics to determine how many switches to perform between recording samples to ensure that the successive samples will be sufficiently statistically independent. In this sense, these algorithms are not exact.

\section{Building a single connected realization of a degree sequence}
\label{sec:connected}

In this section we present a new simple and elegant algorithm to build a connected realization of a degree sequence, if one exists. Constructing such a graph is the first step of any edge-switching MCMC sampling algorithm. We will show two versions of the construction process: to build either a simple graph, or a loopless multigraph.

Let us first consider constructing an arbitrary, not-necessarily-connected \emph{simple} realization of a degree sequence. The Erd\H{o}s--Gallai theorem provides a fast way to check whether a degree sequence is graphical, but not to construct a corresponding graph. To build a realization of the degree sequence, we can use the well-known Havel--Hakimi theorem.

\begin{thm}[Havel and Hakimi, \cite{Havel1955,Hakimi1962}]
\label{thm:hh}
The degree sequence $\mathbf{d} = (d_1 = \Delta, d_2 \ge d_3 \ge \cdots \ge d_n)$ is graphical if and only if after connecting vertex 1 to the $\Delta$ largest-degree vertices, the remaining degree sequence $\mathbf{d}' = (d_2-1, d_3 - 1, \dots, d_{\Delta+1}-1, d_{\Delta+2}, \dots, d_n)$ is also graphical.
\end{thm}

\noindent
This theorem can be understood as an algorithm to construct a simple graph: As with the configuration model, we consider the vertices of the graph with as many stubs as their degrees (\myref{figure}{fig:stubs}). In each step of the algorithm, we select an arbitrary vertex (the ``hub''), and connect all of its stubs to the other vertices that have the most unconnected stubs left (highest \emph{remaining degree}). The hub is then dropped from the degree sequence, along with any other vertices that have no remaining stubs. This step is repeated until no more degrees remain, or until no stubs can be connected without forming a non-simple graph. The theorem states that a degree sequence is graphical if and only if after performing a single step of the algorithm on it, the \emph{remaining degree sequence} formed by the yet-unconnected stubs is also graphical. Thus, the algorithm will succeed in connecting up all the stubs if and only if the original degree sequence was graphical to begin with. This provides a way to both check the graphicality of a degree sequence and to build one of its realizations at the same time.

The Havel--Hakimi algorithm can construct a realization of a degree sequence, but how can we construct a \emph{connected} realization? Previously, this has been done by first constructing an arbitrary, not necessarily connected realization, then using appropriately chosen edge switches (\myref{figure}{fig:edge-switch}) to connect together the components of the graph \cite{Viger2015}. This method is complicated and cumbersome to implement.  Here we propose a simple and elegant alternative.

Note that the Havel--Hakimi theorem does not specify which vertex to choose as the hub in each step: any of them will do. Let us refer to choosing the vertex with the smallest remaining degree as a ``HH*-step''.

\begin{thm}[connected Havel--Hakimi]
\label{thm:conn-hh}
Given a graphical degree sequence, the smallest-first Havel--Hakimi algorithm (i.e.~consisting of HH*-steps) will produce a connected graph if and only if the starting degree sequence was potentially connected.
\end{thm}

\begin{proof}
The key to the proof is to show that if the starting degree sequence is potentially connected, then every HH*-step reduces the number of vertices having non-zero remaining degree precisely by one, except in the very last step, when two vertices with remaining degree 1 each are connected to each other to complete the graph.  Reversing the order of the steps would then correspond to building a graph by adding one vertex at a time and connecting it to some existing vertices.  This clearly results in a connected graph.

Let us think about what kind of degree sequence we must apply a HH*-step to in order to reduce the number of vertices by more than one.  The hub vertex is always removed.  Additional vertices will only be removed if they only had one remaining stub (i.e.~they had degree 1), which was then connected up to the hub vertex.  Since we always choose a smallest-degree vertex as the hub, and connect it to the other vertices with the highest degrees, this situation is only possible when both the smallest and largest degree is 1.  For example, the degree sequence $(1,1,1,1)$ is transformed to $(1,1)$ after one HH*-step, i.e.~it decreases in size by 2. Except for $(1,1)$, such degree sequences consisting solely of 1s are not potentially connected. Thus, we have established that a HH*-step removes precisely one vertex from any potentially connected degree sequence of length greater than two.

In the following, we will show that one HH*-step transforms any potentially connected degree sequence into another potentially connected one. Therefore, vertices are removed one at a time throughout the HH* construction procedure, up to the very last step when two degree-1 vertices are connected to each other. This will complete the proof of the theorem.

Note that with an arbitrary graph construction process, it is not necessary to maintain the potential connectedness of intermediate degree sequences in order to arrive to a connected graph. Maintaining potential connectedness at intermediate stages is a sufficient, but not a necessary condition for obtaining a connected graph. To show that the remaining degree sequences stay potentially connected throughout the HH* construction process, we invoke the following lemma:

\begin{lem}[potential connectedness]
\label{thm:pc}
Let $(d_1, d_2, \dots, d_n)$ be the degree sequence of a (not necessarily simple) graph. There is a connected realization of this degree sequence if and only if $\frac{1}{2} \sum_i d_i \ge n-1$ and $d_i \ne 0, \forall i$, or if $n=1$.
\end{lem}

\noindent
The proof is given in \ref{apd:pc}.

Will the inequality required for potential connectedness in \myref{lemma}{thm:pc} stay valid after modifying the degree sequence with a HH*-step? The right-hand side will decrease by $1$ from $n-1$ to $n-2$. If the selected hub vertex had degree 1, then the left-hand-side also decreases by 1, thus the inequality is maintained.

If the hub vertex had degree $\Delta \ge2$, then the sum of degrees is at least $n \Delta$, $\sum_{i=1}^n d_i \ge n \Delta$.  After one HH* step, the sum of degrees decreases by $2 \Delta$, thus we only need to show that $n\Delta/2 - \Delta \ge (n-2)$, which is obviously true for $\Delta \ge 2$. The inequality is maintained again. \end{proof}

Let us now consider the case of loopless multigraphs, which may be constructed with a procedure analogous to the Havel--Hakimi algorithm.

\begin{thm}[loopless multigraph construction]
\label{thm:cons-multi}
The degree sequence $\mathbf{d} = (d_1, d_2 \ge d_3 \ge \cdots \ge d_n)$ is multigraphical if and only if after connecting vertex 1 to vertex 2 with a single edge, the remaining degree sequence $\mathbf{d}' = (d_1 - 1, d_2-1, d_3, \dots, d_n)$ is also multigraphical.
\end{thm}

\noindent
In simpler terms, in order to construct a loopless multigraph, we may simply select an arbitrary vertex and connect it to a largest-degree one among the other vertices. Repeating this step results in a loopless multigraph if and only if the starting degree sequence was multigraphical. Unlike in the case of the Havel--Hakimi theorem, connections are made one edge at a time.

\begin{proof}
Clearly, if $\mathbf{d'}$ is multigraphical, then so is $\mathbf{d}$. Thus we need only show that the multigraphicality condition of \myref{theorem}{thm:multigraphical}, $\frac{1}{2} \sum_i d_i \ge \dmax$, is maintained after adding a connection between a maximal degree vertex and another vertex. Adding one connection decreases the left-hand-side of the inequality by 1. For the right-hand-side, there are two cases: (1)~If only one vertex had maximal degree, or if precisely two vertices had maximal degree and they were connected to each other, then the right-hand-side (i.e.\ \dmax) also decreases by 1, and the inequality is maintained. (2)~If there is more than one maximal degree vertex and the connection was made between a maximal degree and a non-maximal-degree vertex, then \dmax\ does not decrease. However, in this case, the sum of degrees in $\mathbf{d}$ includes \dmax\ twice, and at at least one more positive term due to the non-maximal-degree vertex. Therefore, $\sum_i d_i > 2\dmax \Leftrightarrow \sum_i d_i \ge 2(\dmax+1)$, so decreasing the left-hand-side by 1 will not violate the inequality.
\end{proof}

\noindent
We can also formulate the analogue of the \myref{theorem}{thm:conn-hh} for the loopless multigraph case:

\begin{thm}[connected loopless multigraph construction]
\label{thm:conn-cons-multi}
Let $\mathbf{d}$ be a multigraphical degree sequence, and let us repeatedly select the largest remaining degree vertex and the smallest non-zero remaining degree vertex, and connect them with a single edge. This construction procedure results in a connected graph if and only if $\mathbf{d}$ was potentially connected.
\end{thm}

\begin{proof}
The proof is completely analogous to that of \myref{theorem}{thm:conn-hh}, and proceeds in three steps: (1)~We will show that after applying a single step of the construction process, the remaining degree sequence stays potentially connected. (2)~Therefore, when applying a single step of the construction process to a potentially connected degree sequence, the number of non-zero remaining degrees decreases by no more than one, except in the very last step. (3)~Consequently, reversing the order of steps constructs a connected graph.

To show that a single construction step keeps the degree sequence potentially connected, we must prove that the condition of \myref{lemma}{thm:pc}, $\frac{1}{2} \sum_i d_i \ge n-1$, is maintained. Since adding a single connection decreases the left-hand-side by 1, this inequality could only be violated if $\frac{1}{2} \sum_i d_i = n-1$ and $n$ (the number of non-zero degrees) does not decrease after a connection step. But since a smallest-degree vertex is always connected, this could only happen if none of the degrees are 1, i.e.\ $d_i \ge 2,\forall i$, which would imply that $\frac{1}{2} \sum_i d_i \ge n \Leftrightarrow \frac{1}{2} \sum_i d_i > n-1$.
\end{proof}

\begin{figure}[tb]
   \centering

   \includegraphics{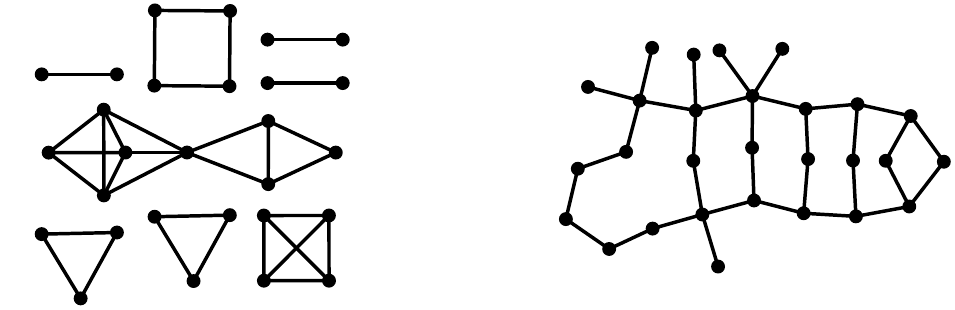} 
   
   \rule{16em}{0.5pt}
   \vspace{0.7em}
   
   \includegraphics{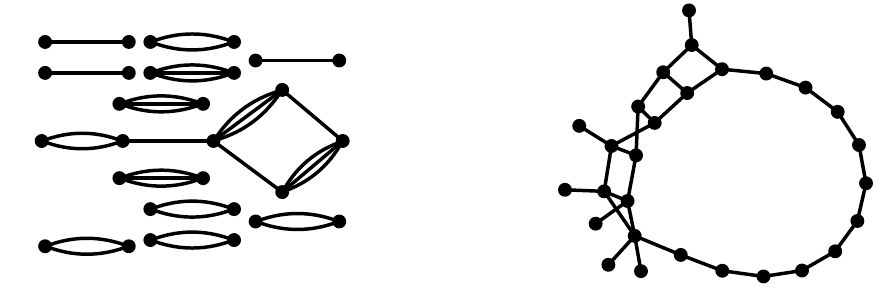} 
   
   \vspace{0.5em}   
   
   {\small $\mathbf{d} = (5, 4, 4, 4, 3, 3, 3, 3, 3, 3, 3, 2, 2, 2, 2, 2, 2, 2, 2, 2, 2, 2, 1, 1, 1, 1, 1, 1)$}
   
   \caption{Top: Two realizations of the same degree sequence constructed with the Havel--Hakimi algorithm, using two strategies for selecting the hub vertex in each step: Selecting the largest degree vertex leads to many components (left) while selecting the smallest degree vertex leads to a connected graph according to \myref{theorem}{thm:conn-hh} (right). Bottom: The same using the loopless multigraph construction procedure of \myref{theorem}{thm:cons-multi}.}
   \label{fig:conn-hh}
\end{figure}

There is a simple intuition behind the statements of theorems \ref{thm:conn-hh} and \ref{thm:conn-cons-multi}. If we were to always choose the highest degree vertex as the hub,
and connect it to other highest-degree vertices, it would quickly use up the available stubs. There would be insufficient stubs left towards the end of the procedure to connect all components together. Indeed, choosing highest-degree vertices as the hub tends to create graphs with multiple dense components (see \myref{figure}{fig:conn-hh}). Conversely, choosing smallest-degree vertices as the hub and connecting them to highest-degree vertices leaves free stubs available. 
The same intuition raises the question: does the largest-first variant of the algorithm always build a non-connected realization, if one exists? The answer turns out to be no. A counterexample is $\mathbf{d} = (3,2,2,2,2,2,1)$, which can be split into two graphical degree sequences $(3,2,2,1)$ and $(2,2,2)$, therefore it has a non-connected realization. Yet the largest-first Havel--Hakimi algorithm can only construct a connected one as it must connect the vertex of degree 3 to three degree-2 vertices. To the best of our knowledge, finding the computational complexity of deciding whether a degree sequence has a non-connected realization as a simple graph, i.e.\ whether it is \emph{forcibly connected}, is still an open problem.  We are not aware of any polynomial-time solutions.  An exponential time algorithm was given by Wang \cite{Wang2018}.

We have contributed an implementation of the construction algorithms for connected simple graphs and connected loopless multigraphs to the igraph C library \cite{Csardi2006} as \texttt{igraph\_{\allowbreak}realize\_{\allowbreak}degree\_{\allowbreak}sequence()}, and made it conveniently accessible through igraph's \textit{Mathematica} interface, IGraph/M \cite{IGraphM}, as the \texttt{IGRealizeDegreeSequence} function. In python-igraph it will be available as the \texttt{hh} method of \texttt{Graph.Degree\_Sequence}.

\section{An exact biased sampling method}
\label{sec:biased}

Recently, a new family of stub-matching sampling methods was proposed \cite{Kim2009,DelGenio2010,Kim2012,Blitzstein2010,Bassler2015}, which construct each sample directly and independently (unlike edge-switching MCMC methods) and work efficiently in polynomial time (unlike the configuration model). These algorithms do not sample uniformly, but they can compute the exact probability of obtaining a sample at the same time as generating that sample. This makes it possible to ``unbias'' the samples and estimate any property that characterizes the entire set of realizations of a degree sequence, such as the averages of various graph metrics, similarly to how one might do with uniform sampling. Let $\mathcal{S} = \{G_1, G_2, \dots, G_K \}$ be the set of generated samples, and let $c(G)$ denote some numerical property of the graph $G$, such as its diameter, assortativity, clustering coefficient, etc. If the sampling is uniform, we can estimate the average of $c$ over all realizations as 
\begin{equation}
\langle c \rangle \approx \frac{1}{K} \sum_{i=1}^K c(G_i) .
\end{equation}
If the sampling is biased, i.e.\ some graphs are generated with a higher probability $p(G)$ than others, then we can re-weight them with $1/p(G)$ to estimate $\langle c \rangle$ as
\begin{equation}
\langle c \rangle \approx 
\left. \Biggl( \sum_i \frac{c(G_i)}{p(G_i)} \Biggr)
\middle/
\Biggl( \sum_i \frac{1}{p(G_i)} \Biggr)
\right.\!\!.
\label{eq:unbias}
\end{equation}
The same formula can be used if we do not have normalized probability values, but merely \emph{sampling weights} $w(G) \sim p(G)$ which are proportional to the probabilities. This is the same principle as the one used in importance sampling.


To illustrate how this class of sampling methods works, let us consider the configuration model again, which pairs the stubs randomly.  Along the same lines, we can exhaustively generate \emph{all realizations} of a degree sequence by connecting up the stubs in all possible ways. This procedure can be thought of as a tree of decisions, like the one shown in \myref{figure}{fig:dectree}: If there are $k=\sum_i d_i$ stubs in total, there will be $k-1$ choices for connecting up the first stub. This is represented by the $k-1$ branches of the tree starting from its root. In the next step (corresponding to the next level of the tree), there will be $k-3$ choices, then $k-5$, and so on.  The leaves of the decision tree represent the fully constructed graphs.

\begin{figure}[tb!]
   \centering
   \includegraphics{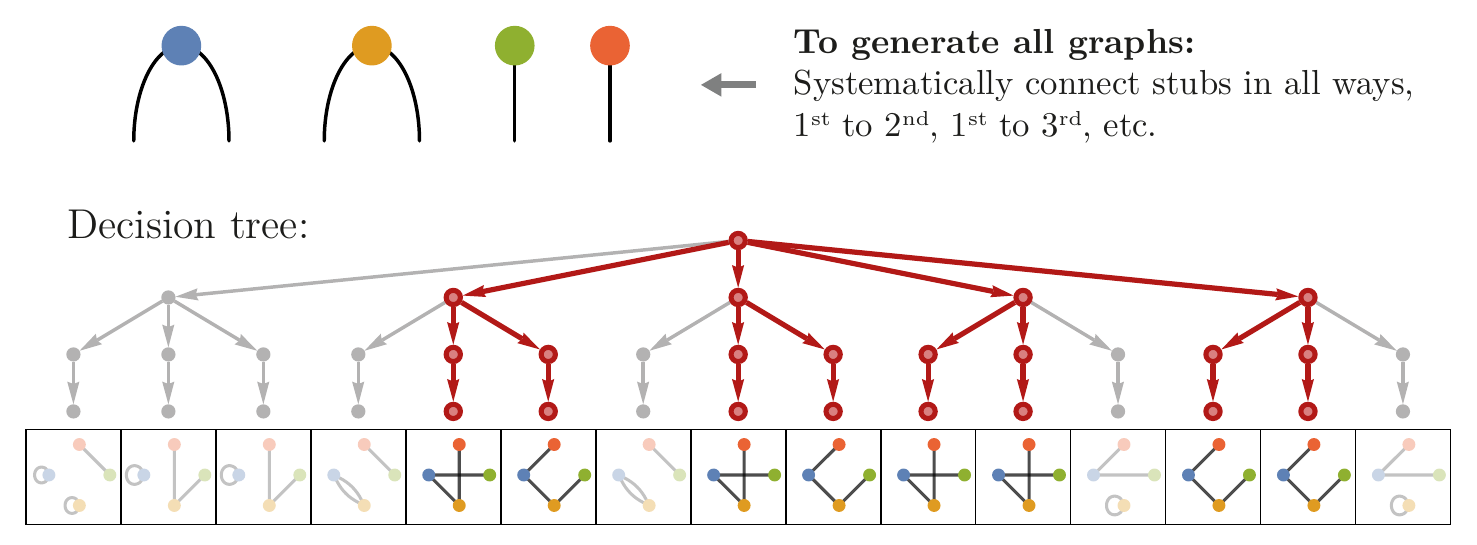} 
   \caption{The decision tree for connecting up the stubs corresponding to the degree sequence $(2,2,1,1)$ in all possible ways. The leaves of the tree represent the labelled graphs that can be obtained by this construction. Only the branches highlighted in red (the \emph{feasible subtree}) lead to simple graphs.}
   \label{fig:dectree}
\end{figure}

The configuration model's algorithm can be thought of as traversing the decision tree randomly, starting at its root, choosing branches uniformly at random at each branching point, and finally arriving at a leaf. This decision tree is symmetric: all tree nodes $i$ steps away from the root (i.e.\ at level $i$ of the tree) have the same number of branches, $k-(2i+1)$. Therefore, each leaf is reached with the same probability $p = \frac{1}{k-1} \times \frac{1}{k-3} \times \cdots = \frac{1}{(k-1)!!}$, where $n!! = n (n-2) (n-4) \cdots$~denotes the double factorial. While each labelled graph appears as more than one tree leaf, all \emph{simple} realizations appear the same number of times, with multiplicity $\prod_i (d_i!)$. This explains why the configuration model samples uniformly if non-simple outcomes are rejected. If we admit loopless multigraph outcomes as well, then the number of leaves that a graph appears as decreases by a factor of $\prod_{i<j} (a_{ij}!)$, where $a_{ij}$ denotes the number of edges between vertices $i$ and $j$ \cite{NewmanBook}.

The part of the decision tree that leads to simple graphs is highlighted in red in \myref{figure}{fig:dectree}. The core idea behind this new class of sampling methods is to traverse only this \emph{feasible subtree}. The feasible subtree is not, in general, symmetric, therefore its leaves will not be sampled uniformly. However, the inverse sampling weight of a leaf can be computed by multiplying the number of feasible branches at each branching point on the path going from the tree root to the leaf. If not all graphs appear as the same number of leaves (as is the case with multigraphs), then the sampling weights used in Eq.~\eref{eq:unbias} must be divided by the appropriate multiplicity.
Through this approach, it is straightforward to take any algorithm that systematically generates all realizations of a degree sequence, and convert it into a random sampling algorithm. Instead of traversing all branches in its decision tree, simply pick a random branch to follow at every step. In order for such an algorithm to be efficient and practical, the following requirements must be met: (1)~the multiplicity of each graph, i.e.\ the number of leaves that correspond to it, must be computable (2)~it must be possible to count the feasible branches at each branching point, and select one of them efficiently. We note that depending on the exhaustive generation algorithm that the sampling is based on, it may be the case that some leaves of the decision tree that correspond to the same graph will have different sampling weights. However, Eq.~\eref{eq:unbias} is still valid for estimating population averages.

A natural generalization of this method is to choose decision branches non-uniformly. This gives an opportunity to reduce the bias of the sampling. If each branch were chosen with probability proportional to the number of leaves it contains, then the sampling would be uniform. While computing the exact number of leaves is a difficult combinatorial problem that may not be efficiently solvable, the sampling can be improved through heuristic choices of the branch probabilities. This idea is explored in more detail in \cite{BasslerNew}. For all the numerical examples discussed in \myref{section}{sec:numerics}, we weighted the branches of the decision tree using a simple heuristic that is described in \ref{apd:heur}.

Here, we choose to work with the decision tree of the exhaustive generation algorithm described above and illustrated in \myref{figure}{fig:dectree}: take the stubs one-by-one, in order, processing all stubs of a vertex before moving on to the next, and consider all possible ways each stub can be connected. This decision tree has $O(m)$ branches at each branching point of each level, where $m = \frac{1}{2} \sum_{i=1}^n d_i$ denotes the number of edges in the constructed graph.  Since the stubs of a vertex are indistinguishable, only $O(n)$ of these branches are distinct. Thus, enumerating each individual branch and testing it for feasibility becomes computationally tractable. Since there are $O(m)$ levels in the tree, the sampling algorithm performs $O(n m)$ feasibility checks during the construction of a graph. For each branch, we must perform two checks: one of graphicality (or multigraphicality) and one of potential connectedness. In the following, we show that both of these checks can be done in constant computational time on average. Therefore, in summary, the computational time required to generate one sample is $O(n m)$, where $n$ is the number of vertices and $m$ is the number of edges of the generated graph.

\textbf{The constraint of graphicality.} When examining the feasibility of a branch, first we must determine if it leads to any simple graphs. This check is similar to the usual graphicality test, with an important difference: Suppose that some stubs of vertex $i$ (the ``hub vertex'') have already been connected to vertices $X = \{j_1, \dots, j_k\}$, but it still has some free stubs. In order to obtain a simple graph, a second connection is not allowed to the vertices in the set $X$. This restriction is referred to as a \emph{star constraint} on $i$, as the connections from $i$ to the elements of $X$ form a star graph. To check graphicality under this constraint, we use the following theorem:

\begin{thm}[star-constrained graphicality \cite{Kim2009}]
\label{thm:scg}
Let $\mathbf{d} = (d_1 = \Delta, d_2 \ge d_3 \ge \cdots \ge d_n)$ be a degree sequence and let $X = \{j_1, j_2, \dots, j_k\}$, with $k \le n-1-\Delta$, be an ``exclusion set'' of vertices to which we forbid connections from vertex 1. Let us connect all stubs of vertex 1 to the $\Delta$ largest-degree vertices not present in $X$, obtaining the remaining degree sequence $\mathbf{d}'$. The degree sequence $\mathbf{d}$ can be realized by a simple graph respecting the exclusion set $X$ if and only if $\mathbf{d}'$ is graphical.
\noclub[2] 
\end{thm}

\noindent
Notice that this is a generalization of the Havel--Hakimi theorem, which corresponds to the special case of $X=\varnothing$, i.e.\ no exclusion. 
The graphicality of $\mathbf{d'}$ can be tested using the Erd\H{o}s--Gallai theorem, making the entire test possible in $O(n)$ computational time. In principle, \myref{theorem}{thm:scg} could be used to test each branch of the decision tree separately, but this would not be efficient. A more sophisticated method is presented in \cite{DelGenio2010}, where it is shown that there exists a threshold degree $d_\text{th}$ that separates feasible branches from non-feasible ones. Connecting to a vertex with degree $d \ge d_\text{th}$ preserves graphicality while connecting to one with $d < d_\text{th}$ does not. $d_\text{th}$ may be determined in $O(n)$ time, thus testing the graphicality of individual branches becomes constant time on average. For a detailed description of this testing procedure, we refer the reader to \cite{DelGenio2010}.

\textbf{The constraint of multigraphicality.} If we wish to sample loopless multigraphs instead of simple graphs, \myref{theorem}{thm:multigraphical} can be used directly. This requires computing the degree sum, as well as the maximum degree. Instead of recomputing these quantities at each step, their values can be updated incrementally in amortized constant time after the addition of each new edge.

\begin{figure}[bt]
   \centering
   \includegraphics{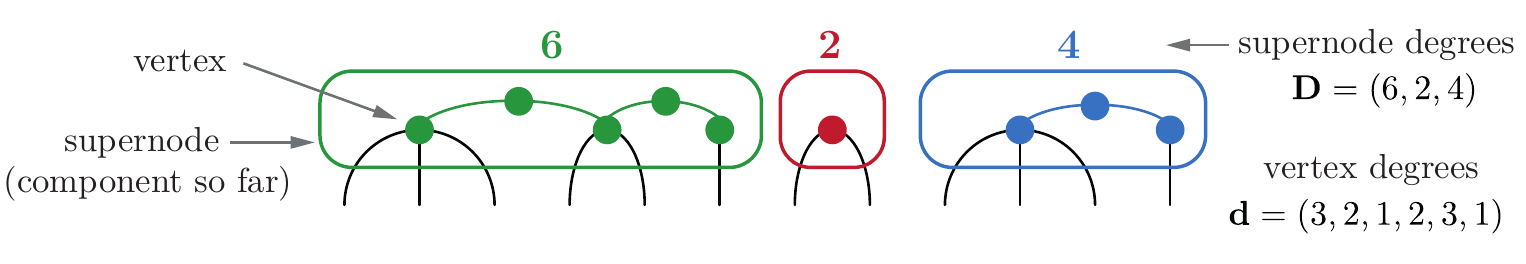} 
   \caption{As the wiring algorithm proceeds, sets of vertices that have already formed a connected component are grouped into ``supernodes''. Whether the remaining stubs can be wired up so as to make the entire graph connected can be decided by applying the potential connectedness theorem to the degree sequence of the supernodes (\myref{lemma}{thm:pc2}).}
   \label{fig:supernodes}
\end{figure}

\textbf{The constraint of connectedness.}  In order to incorporate the constraint of connectedness, we must find a way to detect decision branches which do not lead to any connected graphs. In other words, we must detect when adding a specific connection would make it impossible to build a connected graph. We do this by tracking the groups of vertices (components) which have so far been connected (\myref{figure}{fig:supernodes}). These components can also be thought of as the nodes of a multigraph, which we term the ``supergraph''. We refer to the components as ``supernodes''.  Then the construction process can be completed to a connected graph if and only if the supergraph is potentially connected.

The potential connectedness of the supergraph may be checked using \myref{lemma}{thm:pc}. Note that the supergraph does not need to be a simple graph, and indeed it is clear that when the supernode degrees are sufficiently large, it cannot be simple. It is not obvious that in such a situation it can be ensured that the graph of vertices is simple (or loopless) and the supergraph is connected at the same time. The following lemma asserts that this is indeed possible:

\begin{lem}[potential connectedness of supernodes]
\label{thm:pc2}
Let $\mathbf{D} = (D_1, \dots, D_N)$ be a degree sequence of $N$ supernodes, and let $\mathbf{d} = (d_1, d_2, \dots, d_n)$ be the degree sequence of the vertices making up these supernodes. If $\mathbf{d}$ is graphical (resp.\ multigraphical), the number of edges satisfies $m = \frac{1}{2} \sum_i D_i \ge N-1$ and $D_i \ne 0, \forall i$ or $N=1$, then there is a simple graph (resp.\ loopless multigraph) realization of $\mathbf{d}$ in which the supernodes form a connected graph.
\end{lem}

\noindent
The proof is given in \ref{apd:pc}.

Note that there are two ways in which the potential connectedness of the supergraph can be broken: (1)~there may no longer be a sufficient number of edges left to make the graph connected, i.e.\ $m < N-1$, or (2)~one of the supernodes (components) may become ``closed'', i.e.\ its degree may become zero before the graph is fully constructed. To check whether adding a connection would give rise to either of these two conditions, we must consider several cases: If there is only one supernode, then the graph is already connected, therefore all connections are allowed. Otherwise, if $m = N-1$, then only connections between different supernodes are allowed. Two supernodes with degree 1 each may not connect to each other, and a supernode with degree 2 may not connect to itself, except as the very last step that completed the graph. To check for these cases, we must determine if the two vertices to be connected are within the same supernode. This can be done in constant amortized time, as described in \ref{apd:conn}.

\section{Numerical results}
\label{sec:numerics}

To demonstrate the practical applicability of our proposed sampling method for connected graphs, we performed numerical experiments on degree sequences sampled from a power-law distribution. Networks with similarly heavy-tailed degree distributions commonly occur in the real world \cite{Voitalov2019,Broido2019}.  The exponent of the power-law distribution was adjusted so as to obtain a degree sequence which, while potentially connected, has overwhelmingly many non-connected realizations. Sampling its connected realizations is therefore not feasible at all with the configuration model: in practice it never generates any connected samples. Thus, we compare results with MCMC samplers.

\begin{figure}[htb]
   \centering
   \includegraphics{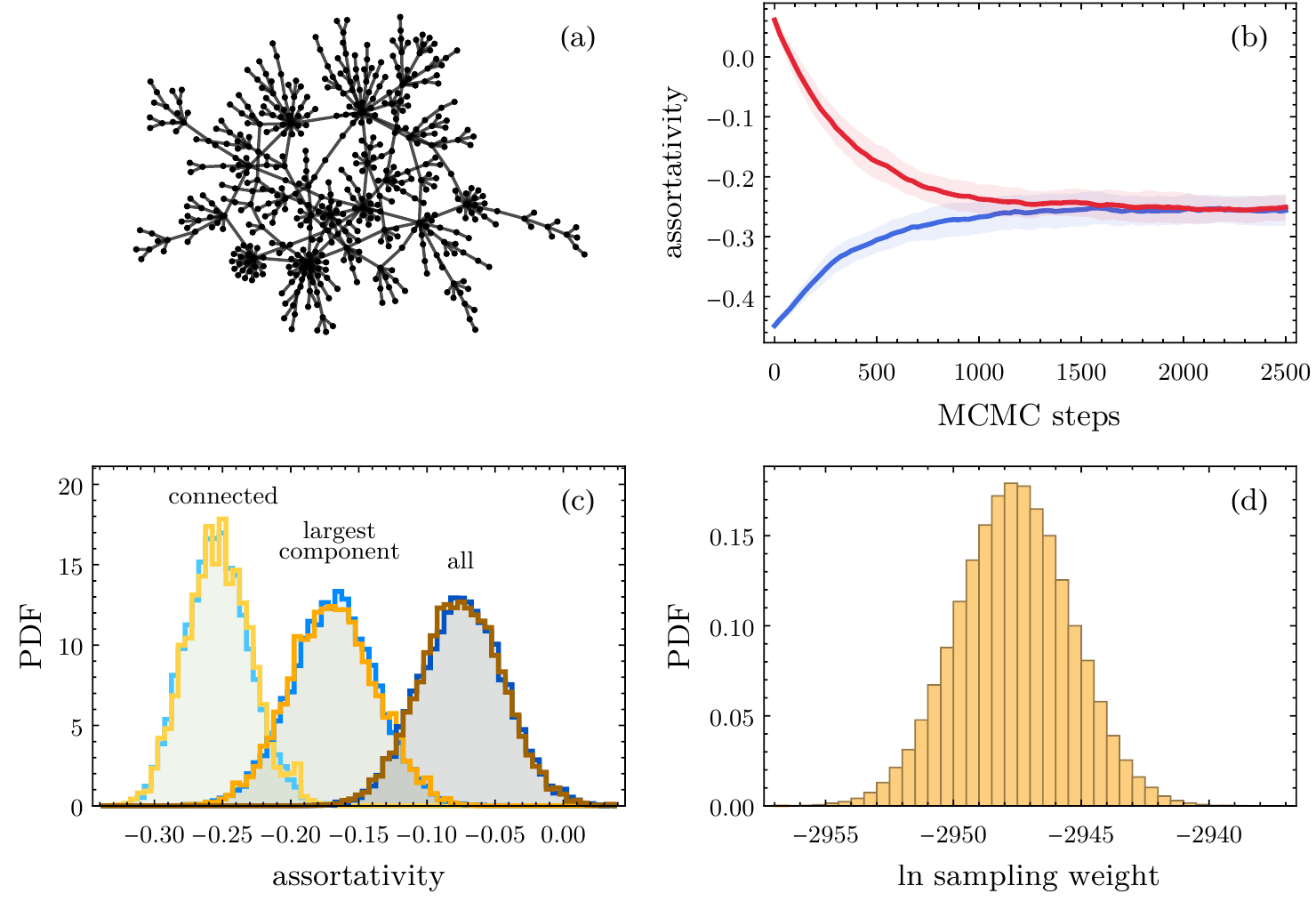} 
   \caption{\textbf{(a)} A typical connected realization of the degree sequence used in the following panels (500 vertices, 519 edges). The degrees were sampled from a power-law distribution with exponent $1.3$. \textbf{(b)} Evolution of a high- and a low-assortativity realization of the degree sequence during repeated applications of random connectivity-preserving edge switches. \textbf{(c)} The distribution of assortativity values is markedly different when sampling from connected realizations, sampling from all realizations, or sampling from all realizations and taking the largest component. The blue histograms were obtained with MCMC samplers
   while the yellow/brown ones with the biased stub-matching samplers. 
   \textbf{(d)} The distribution of the natural logarithms of the sampling weights when using the connected biased sampler.}
   \label{fig:example}
\end{figure}

\myrefp{Figure}{fig:example}{a} shows one typical simple connected realization of such a degree sequence. This degree sequence was used to generate the results shown in the subsequent panels of the same figure. We chose \emph{assortativity}, a measure of degree correlations \cite{Newman2002}, as the graph property to study. \myrefp{Figure}{fig:example}{b} illustrates how the value of this measure develops while running an edge-switching MCMC sampler for simple connected graphs. Two trajectories are shown: one starting with a high- and one with a low-assortativity graph. In this experiment, at least 1500-2000 edge switches were needed before the two trajectories converged, an indicator of reaching statistical independence. Based on this, in the following numerical experiments 2500 steps were performed between taking samples from the Markov chain. In general, the number of steps which are required to guarantee a given level of independence cannot be determined exactly---this is precisely the problem that the biased stub-matching sampler introduced in this work is meant to overcome.

\myrefp{Figure}{fig:example}{c} compares the distribution of assortativity estimated using the MCMC sampler (blue curves) with the one obtained using the biased stub-matching sampler (yellow/brown curves), and demonstrates that both methods produce the same result. This validates our implementation of the method. The histogram of a biased sample is formed not by counting the number of data points in each bin, but by adding up their inverse sampling weights. The result shown in panel \myrefp{figure}{fig:example}{c} comes from three separate experiments: In the first, only connected realizations were sampled. In the second, connectedness was not constrained. In the third, connectedness was also not constrained, but assortativity was measured only on the largest connected component (the ``giant component'') of the graph. We included the third case because retaining only the giant component is often used as an ad-hoc substitute for incorporating the constraint of connectedness into random graph models \cite{Viger2015}.  The assortativity distributions are markedly different for all three cases, demonstrating the importance of taking connectedness into account when the problem at hand demands it. We note that with some degree distributions, simply taking the giant component of non-connected samples produces results similar to enforcing connectedness. However, as \myrefp{figure}{fig:example}{c} demonstrates, with some other degree sequences there can be a significant difference.

\begin{table}[tb]
   \centering
   \footnotesize
   \begin{tabular}{lSSSS}
      \toprule
      &
      \multicolumn{2}{c}{Connected realizations} &
      \multicolumn{2}{c}{All realizations} \\ 
      \cmidrule(lr){2-3}
      \cmidrule(lr){4-5}

      &
      \multicolumn{1}{c}{MCMC}   &
      \multicolumn{1}{c}{biased} &
      \multicolumn{1}{c}{MCMC}   &
      \multicolumn{1}{c}{biased} \\

      \midrule
      mean      & -0.2558 \pm 0.0003 & -0.2550 \pm 0.0008  
                & -0.0769 \pm 0.0003 & -0.0777 \pm 0.0003 \\
      std.~dev. &  0.0236 \pm 0.0002 &  0.0241 \pm 0.0007 
                &  0.0308 \pm 0.0002 &  0.0306 \pm 0.0002 \\
      skewness  &  0.18   \pm 0.02   &  0.17   \pm 0.07  
                &  0.01   \pm 0.02   &  0.00   \pm 0.02 \\
      kurtosis  &  2.93   \pm 0.05   &  3.03   \pm 0.10  
                &  2.90   \pm 0.04   &  2.96   \pm 0.05  \\
      \bottomrule
   \end{tabular}
   \caption{The first four statistical moments of the assortativity distributions shown in \myrefp{figure}{fig:example}{c}, as estimated with MCMC and with the biased stub-matching sampler. Standard errors obtained with bootstrapping and are indicated in parentheses.}
   \label{tab:moments}
\end{table}

Estimates of four statistical moments of the distributions---their mean, standard deviation, skewness and kurtosis---are reported in \myref{table}{tab:moments} along with their standard errors. We note that the number of samples required for an accurate estimate of statistical quantities is larger when using biased sampling than with uniform sampling. This is not dissimilar from how the effective sample size of the correlated output of an MCMC sampler is also smaller than the number of generated data points. Therefore, when generating the histograms in \myrefp{figure}{fig:example}{c}, we took 10\,000 samples from the Markov chain (at intervals of 2500 steps) and 100\,000 samples from the biased sampler. In the case of the biased sampler described here, the distribution of sample weights is typically bell-shaped on a logarithmic scale, as shown in \myrefp{figure}{fig:example}{d}. This is expected, since sample weights are the inverse products of the number of feasible branches encountered at each level while traversing the decision tree. If the number of branches were random, the distribution of weights would be log-normal according to the central limit theorem.

Finally, as an example application of the method, we investigate the properties of two connected real-world networks by comparing them to a null model with degree and connectedness constraints (\myref{figure}{fig:emp}). Both of these networks are sufficiently sparse so that most realizations of their degree sequences are disconnected. Therefore, the connectedness constraint cannot be handled with simple rejection. The first network is the equivalenced representation of the Western US power network, from the Harwell--Boeing sparse matrix collection (443 vertices, 590 edges) \cite{PowerGrid}. As a power grid, it is naturally connected. We investigate its \emph{global efficiency}, defined as the average of the inverse of pairwise shortest path lengths between its vertices \cite{Latora2001}. \myrefp{Figure}{fig:emp}{a} shows that the efficiency of this network is significantly lower than that of typical realizations of its degree sequence. This hints at the existence of another dominant constraint, which we surmise to be the spatially embedded nature of the network.  Typical connected realizations have higher efficiency than non-connected ones. As the second example, we investigate the degree assortativity in the largest connected component of the protein-protein interaction network of the yeast \emph{Saccharomyces cerevisiae} (964 vertices, 1487 edges) \cite{Yu2008,YeastInteractome}.  Random networks with the same degrees become more disassortative (have higher negative assortativity) when forced to be connected, but still not as disassortative as the empirical network. This shows that high disassortativity is a special property of this network.

\begin{figure}[tb]
   \centering
   \includegraphics{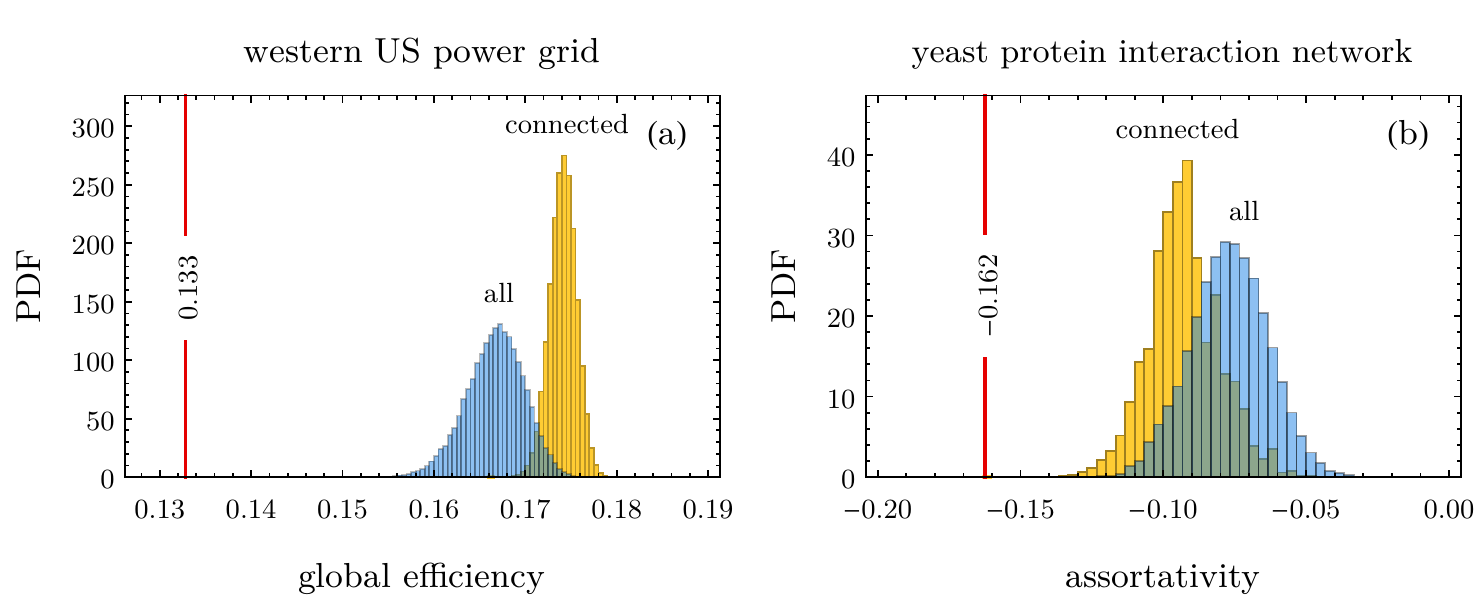} 
   \caption{The distribution of the values of selected graph measures (global efficiency and assortativity) in either the connected realizations (yellow) or in all realizations (blue) of the degree sequences of two connected real-world networks. The vertical red line shows the value of the graph measure for the empirical network itself. The histograms were created via importance sampling.}
   \label{fig:emp}
\end{figure}

\section{Discussion}
\label{sec:discussion}

In this paper we considered the problem of constructing a single realization of a \emph{connected} graph with a given degree sequence, as well as random sampling from the set of all connected realizations. We addressed both the case of simple graphs, as well as loopless multigraphs. The main contribution of this work is incorporating the constraint of connectedness.

Building a not-necessarily-connected realization of a degree sequence as a simple graph can be accomplished using the well-known Havel--Hakimi construction. Until now, the usual method to construct a connected realization was to first build an arbitrary realization, then rewire its edges to make it connected. This method is complicated and cumbersome to implement. With \myref{theorem}{thm:conn-hh}, we show that a specific variant of the Havel--Hakimi construction is guaranteed to produce a connected realization, if one exists. Furthermore, in \myref{theorem}{thm:conn-cons-multi} we generalize the construction to the case of loopless multigraphs. This provides a simple and elegant algorithm for building connected graphs with given degrees. We contributed an implementation of these algorithms to the open-source igraph library \cite{Csardi2006} and its \textit{Mathematica} interface, IGraph/M \cite{IGraphM}.

We have also extended a new family of biased-sampling stub matching methods so that they incorporate the constraint of connectedness without a performance penalty, allowing for fast, efficient rendering of null models and random sampling. Indeed, our approach is significantly faster than the configuration model, which is simply infeasible to use in some regimes of degree sequences. Our algorithm generates each sample in computational time $O(nm)$, where $n$ is the number of vertices and $m$ is the number of edges. Unlike edge-switching MCMC methods, the mixing time of which are not currently known, our method suffers no uncertainty or ambiguity in the independence of the samples. In this sense it is exact. This is of particular importance, again, for the rendering of reliable null models that faithfully represent generic networks of a certain type. An implementation of our sampling method is made freely available at \url{https://github.com/szhorvat/ConnectedGraphSampler}. Finally, we have demonstrated these methods both on generated scale-free degree sequences, as well as on degree sequences of real-world networks. The connected realizations of all of these are markedly different from the non-connected ones, illustrating the relevance of the connectedness constraint. This is consistent with earlier approximate results obtained with heuristic samplers whose bias was not controlled \cite{Ring2020}.  In all these examples, the use of the configuration model would have been simply infeasible.

We reiterate that these approaches are crucially important due to the pressing need for efficient, appropriate null models across the network and complexity sciences. While the general problem of multi-constraint null model construction and random sampling in random graph models remains open, connectedness is such a ubiquitous feature of real networks and graphs of potential interest that we hope our simple and powerful approach to building connected null models and performing random sampling will find wide applicability. Ultimately, reaching a state in which validation of new findings against numerical control experiments is the standard must be a critical goal for the field as a whole, and further progress in multi-constraint sampling is the only way forward.

\ack 

We thank Christoph Zechner, Benjamin Friedrich, \'Agnes T\'oth-Petr\'oczy and Steffen Rulands for helpful comments on the manuscript. Additionally, Sz.~Horv\'at would like to acknowledge useful discussions with Kevin Bassler, Zolt\'an Toroczkai and \'Eva Czabarka.

\appendix

\renewcommand\thesection{Appendix~\Alph{section}} 

\renewcommand\thefigure{A\arabic{figure}}

\section{Proofs of theorems}
\label{apd:proofs}

\subsection{Multigraphicality}
\label{apd:multigraphical}

\begin{proof}[Proof of \myref{theorem}{thm:multigraphical}]
In order to avoid self-loops, no degree may be larger than the sum of all other degrees: $\sum_{i\ne k} d_i \ge d_k, \, \forall \, 1 \le k \le n$. This condition can be rewritten as $\left( \sum_{i=1}^n d_i \right) - d_k \ge d_k, \forall k \Leftrightarrow \sum_{i=1}^n d_i \ge 2d_k, \forall k \Leftrightarrow \sum_{i=1}^n d_i \ge 2\dmax$. Therefore, $\frac{1}{2} \sum_{i=1}^n d_i \ge \dmax$ is a necessary condition of multigraphicality. Now we show it is also sufficient, assuming that the degree sequence has an even sum. Any even-sum degree sequence can be realized as a graph that contains self-loops. Let us consider such a loopy realization and assume that vertex $v$, with degree $d_v$, has at least one self-loop. This self-loop may be eliminated by performing a degree-preserving edge switch (\myref{figure}{fig:edge-switch}) between it and another edge that is not incident to $v$. Such a non-incident edge exists because the total number of edges in the graph is $\frac{1}{2}\sum_i d_i \ge d_v$, while $v$ has at most $d_v - 1$ incident edges due to its self-loop. Therefore, all self-loops can be eliminated using a sequence of degree-preserving edge switches.
\end{proof}

\subsection{Potential connectedness}
\label{apd:pc}

\begin{proof}[Proof of \myref{lemma}{thm:pc}]
The condition of potential connectedness and its proof are well-known, and generally found in graph theory textbooks. We reproduce the proof here because of its relevance to \myref{lemma}{thm:pc2}, which is a generalization of \myref{lemma}{thm:pc}. The requirement that there be no zero-degree (i.e.\ isolated) vertices whenever there is more than one vertex is trivial, so from here on we assume that all degrees are positive. The number of edges in a graph with degree sequence $\mathbf{d} = (d_1, \dots, d_n)$ is $m = \frac{1}{2} \sum_{i=1}^n d_i$, thus we must show that potential connectedness is equivalent to having $m \ge n-1$ edges. First we show necessity. The smallest connected graph on $n$ vertices is a tree, and has precisely $m-1$ edges. This is because any tree with $n>1$ can be extended from a smaller tree by connecting a vertex to it with a single edge. Similarly, it is easy to see that a connected graph with more than $n-1$ edges must contain an edge whose removal does not disconnect it, called a \emph{cycle edge}.

Now we show that any non-connected graph with at least $n-1$ edges can be transformed into a connected one by a sequence of degree-preserving edge switches (\myref{figure}{fig:edge-switch}). Let us assume that we have a realization of $\mathbf{d}$ with two connected components, having vertex counts $n_1,\, n_2$ and edge counts $m_1,\, m_2$. Since $m_1 + m_2 = m \ge n-1 = n_1 + n_2 - 1$, and since both components are connected implying $m_1 \ge n_1 - 1$ and $m_2 \ge n_2 - 1$, one of the two must satisfy $m_1 > n_1 - 1$ or $m_2 > n_2 - 1$. Thus, one of the two components has a cycle edge. Switching this cycle edge with an arbitrary edge in the other component will connect the two components together, without creating multi-edges. Therefore, given a non-connected realization of a degree sequence satisfying $m \ge n-1$, we can connect its components together pair-by-pair, using degree-preserving edge switches. This concludes the proof of sufficiency.
\end{proof}

\noindent
\emph{Remark.} Potential connectedness is independent from graphicality and multigraphicality. These two properties do not constrain each other in the sense that the existence of a connected realization is the same regardless of whether we consider simple graphs, loopless multigraphs, or loopy multigraphs. This is not true for the related concept of \emph{forcibly connectedness}: a degree sequence that only has connected simple realizations may have non-connected multigraph realizations. An example is $(2,2,2,2)$, which is forcibly connected over simple graphs, but not over multigraphs.
\noclub[2]

\begin{proof}[Proof of \myref{lemma}{thm:pc2} (potential connectedness of supernodes)]
\myref{Lemma}{thm:pc2} is completely analogous to \myref{lemma}{thm:pc} with the difference that connectedness is considered over the supernodes (with degrees $\mathbf{D}$), while graphicality (or multigraphicality) is considered over the vertices (with degrees $\mathbf{d}$). The proof is analogous as well: we consider a simple graph (or loopless multigraph) realization of $\mathbf{d}$. This is also a realization of $\mathbf{D}$, although usually a non-simple one. It can be shown that there is a sequence of edge switches that preserve both $\mathbf{d}$ and $\mathbf{D}$ and make the supergraph connected.
\end{proof}

\section{Edge-switching MCMC methods}
\label{apd:mcmc}

In this appendix we formulate an edge-switching MCMC method for sampling graphs, and show that it samples uniformly. One must take great care when formulating such a method, as uniformity hinges on seemingly small details of the algorithm. In our experience, similar methods are frequently implemented incorrectly, resulting in non-uniform sampling.

The algorithm for propagating the Markov chain is as follows:
\begin{enumerate}[(1)]
\item
Choose two distinct edges $(a,b)$ and $(c,d)$ of the graph $G$, each uniformly at random.
\item
Switch the edge-pair into either $(a,c),\, (b,d)$, or $(a,d),\, (b,c)$ (\myref{figure}{fig:edge-switch}), with probability $\frac{1}{2}$ each, obtaining the graph $G'$. The edge switch may potentially create self-loops (if the edges were not disjoint) or multi-edges.
\item
If $G'$ satisfies our chosen constraints (e.g.\ simple graph, loopless multigraph, connected graph, etc.), choose it as the next state of the Markov chain. Otherwise, the next state will be $G$, the same as the current one.
\end{enumerate}
In order for sampling to be uniform, it is necessary that the Markov chain be irreducible, i.e.\ that all graphs are reachable with edge switches that preserve our chosen constraints. This is true for simple graphs,
loopless multigraphs,
as well as for the connected version of both \cite{Taylor1981}.
Notice that detailed balance is fulfilled because edge switches are reversible. Finally, the number of available edge-pair choices in step~(1) is the same in all states, therefore each choice in each state is made with the same probability. This implies that the transition probability matrix of the Markov chain is symmetric, ensuring that its stationary distribution is uniform.

We note that if in step~(1) only those edge-pairs are selected whose switching would maintain the constraints (i.e.\ keep the graph simple or loopless), preventing the Markov chain from staying in the same state for more than one step, then the sampling may not be uniform. For example, one may be tempted to only choose disjoint pairs of edges, as this would avoid the creation of self-loops and speed up the procedure. It turns out that with this choice, the sampling of simple graphs would still be uniform, but the sampling of loopless multigraphs would not. This is because the number of disjoint edge-pairs in a simple graph only depends on its degree sequence, and can be computed as the total number of edge-pairs minus the number of non-disjoint edge-pairs at each vertex:
\[
\binom{{\frac{1}{2}\sum_i d_i}}{2} -
\sum_{i=1}^n \binom{d_i}{2}.
\]
However, this is no longer true for loopless multigraphs, in which two distinct vertices may be incident to the same (parallel) edge-pair. In a loopless multigraph, the number of disjoint edge-pairs depends not only on the graph's degree sequence, but also on its specific structure. The transition matrix is no longer symmetric in general, and the sampling is not uniform. If in step~(1) we entirely avoid edge-switches that would create a non-simple graph, then the sampling of simple graphs would not be uniform either.

\section{Efficient implementation of connectivity tracking}
\label{apd:conn}

During the process of graph construction that we use for sampling connected graphs, we must track the evolution of connected components, i.e.\ that of the ``supernodes'' and their degrees. Here we describe a possible efficient implementation of connectivity tracking.  

To maintain the supergraph's degree sequence, we need to be able to efficiently check if two vertices belong to the same supernode, i.e.\ if they are in the same connected component, at any point during the graph construction. This information can be encoded into a rooted forest whose nodes are identical to the vertices of our graph. Each tree in the forest corresponds to a component. To determine if two vertices are in the same component, we can simply find the roots of their trees and check if they are the same. This way, each supernode will be represented by a tree root, and any associated information (such as the supernode's degree) can be stored directly in the root. When two vertices of the graph are connected to each other, the forest is updated as follows: If the two vertices were in the same component, nothing needs to be done. If they were in different components, then the two corresponding trees are joined by connecting their roots, and designating one of them as the root of the newly created larger tree.

Finding which component (supernode) a vertex belongs to requires walking up its tree until the root is reached. This takes as many operations as the distance of that vertex from the root. To speed up subsequent queries, we take all vertices on the path from the original vertex to the root, and make them direct descendants of the root. This makes checking if two vertices are in the same component a constant-time operation \emph{on average}. To see why, consider that while building the graph, the component of each vertex will be checked before adding a connection. This reduces the depth of each tree in the forest to one. The subsequent joining of trees can thus never create a tree depth greater than two, i.e.\ no component check will take more than two operations.

\section{A heuristic for weighting the branches of the decision tree}
\label{apd:heur}

We employ two simple heuristics to reduce the bias of the sampling distribution: (1)~re-ordering the degree sequence and (2)~choosing branches of the decision tree non-uniformly. 

Note that the structure of the decision tree depends on the order in which vertices are connected up, i.e.\ the ordering of the degree sequence. We observed empirically that when using the connectedness constraint, an increasing ordering of degrees produces a narrower weight distribution, i.e.\ results in ``more uniform'' sampling. This is consistent with the intuition described in \myref{section}{sec:connected}: connecting small degree vertices to larger degree ones favours creating a connected graph.

\begin{figure}[bt]
   \centering
   \includegraphics{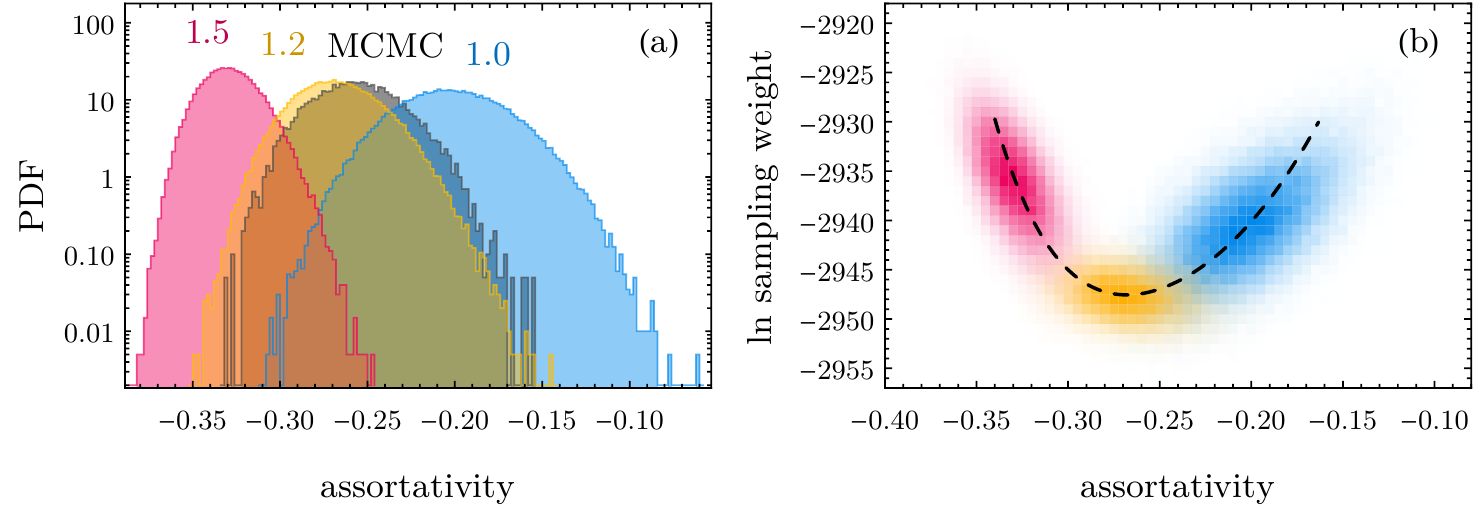} 
   \caption{\textbf{(a)}~Biased distributions of assortativity values obtained with the connected graph sampler for $\alpha=1.0, 1.2, 1.5$ (blue, yellow, red; 100\,000 samples), compared with the unbiased samples obtained with an MCMC sampler (grey; 10\,000 samples). Note the logarithmic vertical scale. \textbf{(b)}~Coloured patches represent the joint distribution of assortativity values and the natural logarithms of sampling weights. The dashed black line is the path traced by the mean of these distributions when varying $\alpha$ between 0.9 and 1.6.}
   \label{fig:apd}
\end{figure}

In the most basic version of the sampling algorithm, each feasible branch of the decision tree is chosen with the same probability, i.e.\ allowed stubs are picked uniformly. This is equivalent to picking vertices with probability proportional to their degrees, $d$. We introduce a simple one-parameter heuristic to choose decision branches non-uniformly: pick vertices with probability proportional to $d^\alpha$, or equivalently, pick stubs with probability proportional to $d^{\alpha-1}$. The parameter $\alpha$ effectively tunes the affinity of connecting to high- versus low-degree vertices. $\alpha=1$ corresponds to uniform stub choice. This choice of weighting the branches of the decision tree is purely heuristic, and is motivated both by its simplicity and the observation that both graphicality and connectedness are affected by a preference to choose larger or small degrees (see sections \ref{sec:connected} and \ref{sec:biased}).  For a more detailed exploration of branch weighting, see \cite{BasslerNew}.

The parameter $\alpha$ must be adjusted to reduce the bias of the sampler as much as possible. We do this based on the observation that the bias manifests itself in two important ways. First, the distribution of the sampling weights (\myrefp{figure}{fig:example}{d}) has a large variance. If sampling were uniform, its variance would be zero. Therefore, $\alpha$ could be chosen so as to minimize the variance of the sampling weight distribution. Second, when measuring a certain graph property such as assortativity, the \emph{biased} sampler may produce property values that should be common with a vanishingly low probability. \myrefp{Figure}{fig:apd}{a} shows the biased distributions of assortativity values obtained with various different choices of $\alpha$ (blue, yellow, red) and compares it to the values obtained with a non-biased MCMC sampler (grey). Notice that the biased distribution obtained with $\alpha=1$ (blue) overlaps with the non-biased one only partially, and, for the sample size used here, includes almost no values lower than $-0.30$. Therefore, the bias cannot be effectively corrected without increasing the sample size significantly. However, the range of values frequently produced by the biased sampler may be adjusted through $\alpha$: increasing $\alpha$ shifts assortativity values to a lower range (\myrefp{figure}{fig:apd}{a}, red and yellow). In the spirit of importance sampling, we choose $\alpha$ to sample ``important'' values with high probability, i.e.\ maximize the overlap of the biased distribution with the non-biased one, and thus minimize the amount of bias correction that is necessary. How can this be achieved without knowing the non-biased distribution a priori? Notice that bias correction will cause a shift in the range of values only if there is a correlation between the values and the sampling weights. \myrefp{Figure}{fig:apd}{b} shows their joint distributions: the correlation is negative for $\alpha=1.5$, positive for $\alpha=1.0$ and mostly vanishes for $\alpha=1.2$. When the distributions are unimodal, as is typically the case, the lowest correlation can be achieved by minimizing the mean logarithmic sampling weight, i.e.\ finding the minimum of the black dashed curve in \myrefp{figure}{fig:apd}{b}. Notice that this may be done without reference to any particular graph measure, such as assortativity. In the examples considered here, we observed that minimizing the mean of the logarithmic sampling weights also reduced their variance.  In summary, minimizing either the variance or the mean of the logarithmic sampling weight distribution are practical ways to improve the performance of the sampling method. 

For all examples presented here, we used the Kiefer--Wolfowitz stochastic ap\-prox\-i\-mation algorithm to find the optimal $\alpha$. The $\alpha$ values used for \myref{figure}{fig:example} were $1.200$ when sampling from the connected realizations of the degree sequence and $1.107$ when sampling from all realizations. The degree sequence was ordered increasingly in both cases.

\section*{References}

\raggedright
\bibliography{degseq}
\bibliographystyle{iopart-num} 

\end{document}